\documentclass[11pt,english,aps,pra,onecolumn,tightenlines,superscriptaddress,notitlepage,floatfix,fleqn]{revtex4-1}
\pdfoutput=1
\usepackage[utf8]{inputenc}
\usepackage[T1]{fontenc}
\usepackage{amsmath}
\usepackage{amssymb}
\usepackage{amsfonts}
\usepackage{bm}
\usepackage{bbm}
\usepackage{epsfig}
\usepackage{grffile}
\usepackage{times}

\usepackage[usenames,dvipsnames]{color}
\definecolor{dblue}{rgb}{0,0.1,.6}

\usepackage[colorlinks=true,citecolor=dblue,linkcolor=dblue,urlcolor=dblue]{hyperref}
\usepackage[all]{hypcap}
\newcommand{\Footnote}[1]{\footnote{\unexpanded{#1}}}

\newcommand{\id}{\mathbbm{1}}

\newcommand{\bra}{\langle}
\newcommand{\ket}{\rangle}
\newcommand{\Tr}{\operatorname{Tr}}
\newcommand{\mc}[1]{\mathcal{#1}}
\renewcommand{\vec}[1]{{\boldsymbol{#1}}}
\newcommand{\pdag}{{\phantom{\dag}}}
\newcommand{\ud}{\mathrm{d}}
\newcommand{\ui}{\mathrm{i}\mkern1mu}
\newcommand{\ue}{\mathrm{e}}

\newcommand{\M}{\mc{M}}
\newcommand{\RR}{\mathbb{R}}
\newcommand{\CC}{\mathbb{C}}
\newcommand{\veps}{\varepsilon}

\newcommand{\hA}{\hat{A}}
\newcommand{\hB}{\hat{B}}
\newcommand{\hU}{\hat{U}}
\newcommand{\hV}{\hat{V}}
\newcommand{\hW}{\hat{W}}
\newcommand{\hH}{\hat{H}}
\newcommand{\hO}{\hat{O}}
\newcommand{\hX}{\hat{X}}
\newcommand{\hY}{\hat{Y}}
\newcommand{\hd}{\hat{d}}
\newcommand{\hg}{\hat{g}}
\newcommand{\heta}{\hat{\eta}}
\newcommand{\hZ}{\hat{Z}}
\newcommand{\tU}{\tilde{U}}
\newcommand{\hUC}{\hat{\mc{U}}}
\newcommand{\hsigma}{\hat{\sigma}}
\newcommand{\dm}{{\hat{\rho}}}

\newcommand{\End}{\operatorname{End}}
\newcommand{\avg}{\operatorname{Avg}}
\newcommand{\var}{\operatorname{Var}}
\newcommand{\cov}{\operatorname{Cov}}
\newcommand{\Avg}{\avg_{\hU}}
\newcommand{\Var}{\var_{\hU}}
\newcommand{\Swap}{\operatorname{Swap}}
\newcommand{\groupU}{\operatorname{U}}
\newcommand{\full}{\text{full}}

\usepackage{amsthm}
\newtheorem{theorem}{Theorem}

\newcommand{\duke} {Department of Physics, Duke University, Durham, North Carolina 27708, USA}
\newcommand{\dqc}  {Duke Quantum Center, Duke University, Durham, North Carolina 27701, USA}
\newcommand{\tz}    {Tensor Center, Auf dem Dresch 15, 52152 Simmerath, Germany}

\begin{document}

\title{Equivalence of cost concentration and gradient vanishing for quantum circuits: An elementary proof in the Riemannian formulation}
\author{Qiang Miao}
\affiliation{\dqc}
\author{Thomas Barthel}
\email{thomas.barthel@duke.edu}
\affiliation{\dqc}
\affiliation{\duke}
\affiliation{\tz}

\begin{abstract}
The optimization of quantum circuits can be hampered by a decay of average gradient amplitudes with increasing system size. When the decay is exponential, this is called the barren plateau problem. Considering explicit circuit parametrizations (in terms of rotation angles), it has been shown in \emph{Arrasmith et al., Quantum Sci.\ Technol.\ 7, 045015 (2022)} that barren plateaus are equivalent to an exponential decay of the variance of cost-function differences. We show that the issue is particularly simple in the (parametrization-free) Riemannian formulation of such optimization problems and obtain a tighter bound for the cost-function variance. An elementary derivation shows that the single-gate variance of the cost function is \emph{strictly equal} to half the variance of the Riemannian single-gate gradient, where we sample variable gates according to the uniform Haar measure. The total variances of the cost function and its gradient are then both bounded from above by the sum of single-gate variances and, conversely, bound single-gate variances from above. So, decays of gradients and cost-function variations go hand in hand, and barren plateau problems cannot be resolved by avoiding gradient-based in favor of gradient-free optimization methods.
\end{abstract}

\date{March 11, 2024}

\maketitle

\section{Introduction}
Recent rapid advancements in quantum computing hardware enable the implementation of large and deep quantum circuits, reaching regimes beyond the simulation capabilities of classical computers. A promising scheme to harness this potential before the advent of practical fault-tolerance are variational quantum algorithms (VQA) \cite{Cerezo2021-3}: Quantum circuits are executed on quantum computers and the quantum-gate parameters  are optimized through a classical backend to minimize a given cost function. A critical challenge in such hybrid quantum-classical optimizations consists in noise and the probabilistic nature of quantum measurements. In generic variational quantum circuits, average gradient amplitudes tend to decrease \emph{exponentially} in the system size (number of qudits). This phenomenon is known as \emph{barren plateaus} \cite{McClean2018-9, Cerezo2021-12}. Unless one already has a very good guess for the optimal circuit, the barren plateau problem implies that we would need an exponential number of measurement shots for a sufficiently accurate determination of cost-function gradients, prohibiting the application for large problem sizes. Otherwise, we would very likely end up with random walks in flat regions of the cost landscape. Numerous works investigate how to avoid an exponential decay of gradient amplitudes~\cite{Grant2019-3,Zhang2022_03,Mele2022_06,Kulshrestha2022_04,Dborin2022-7,Skolik2021-3,Slattery2022-4,Haug2021_04,Sack2022-3,Rad2022_03,Tao2022_05,Wang2023_02,Miao2024-109,Barthel2023_03,Zhang2024-132}, and the absence of barren plateaus might imply classical simulability~\cite{Cerezo2023_12}.
\begin{figure*}[t]
	\includegraphics[width=0.9\textwidth]{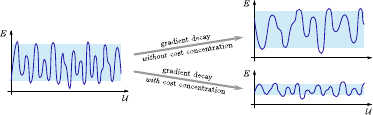}
	\caption{\label{fig:gradConcentration} An increase in dimensionality (scaling up the system size $n$) can lead to a decay of gradients. The situation where the average gradient amplitude decays exponentially in $n$ is the so-called barren-plateau problem. In general, gradient decay may or may not be accompanied by concentration of the cost function. As discussed here and in Ref.~\cite{Arrasmith2022-7} the two phenomena go hand in hand for VQA.}
\end{figure*}

The quantum circuits in VQA can comprise fixed unitary gates $\{\hW_1,\hW_2,\dotsc\}$ and variable unitary gates $\{\hU_1,\hU_2,\dotsc,\hU_K\}$. For example, the former could be CNOT gates and the latter single-qubit gates. The variable gates are typically parametrized by rotation angles, $\hU_i=\hU_i(\vec{\theta}_i)\in\groupU(N_i)$, and the optimization is based on the Euclidean metric in the angle space $\{\vec{\theta}_i\}$.
Using this framework and assuming that the variables gates are composed of rotations $\ue^{-\ui\theta_{i,k}\hsigma_k}$ with involutory generators $\hsigma_k=\hsigma_k^\dag=\hsigma_k^{-1}$, Arrasmith \emph{et al.} established the equivalence of barren plateaus and 
an exponential decay of the variance of cost-function differences with respect to increasing system size \cite{Arrasmith2022-7}.

Alternatively, one can formulate the circuit optimization problem directly over the manifold
\begin{equation}\label{eq:Mfull}
	\mc{M}=\groupU(N_1)\times\groupU(N_2)\times\dots\times\groupU(N_K)
\end{equation}
formed by the direct product of the gates' unitary groups in a representation-free form. In this Riemannian approach \cite{Smith1994-3,Huang2015-25}, gradients are elements of the tangent space of $\M$, and one can implement line searches and Riemannian quasi-Newton methods through retractions and vector transport on $\M$ as discussed in recent works~\cite{Miao2021_08,Wiersema2023-107}. Riemannian optimization has some advantages over the Euclidean optimization of parametrized quantum circuits. For example, it avoids cost-function saddle points that are introduced when employing a global parametrization $\{\vec{\theta}_i\}$ of the manifold $\M$ (consider, e.g., sitting at the north pole of a sphere and rotating around the $z$ axis). Furthermore, the Riemannian formulation can simplify analytical considerations, e.g., concerning average gradient amplitudes \cite{Miao2024-109,Barthel2023_03} and cost-function variances as discussed in the following.
\begin{figure*}[t]
	\includegraphics[width=\textwidth]{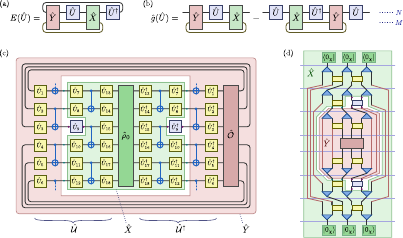}
	\caption{\label{fig:CostAndGrad} Diagrammatic representations for (a) cost function $E(\hU)$ and (b) Riemannian gradient $\hg(\hU)$, where we consider the dependence on a single gate $\hU\in\groupU(N)$ from the variable gates $\{\hU_i\in\groupU(N_i)\}$ that compose the quantum circuit $\hUC$; cf.\ Eqs.~\eqref{eq:costXY} and \eqref{eq:Riem_grad}. Cost functions of this form arise in various applications like quantum machine learning and variational quantum algorithms for the investigation of many-body ground states. Panel (c) shows how a cost function \eqref{eq:cost} for a quantum circuit consisting of single-qubit and CNOT gates attains the form \eqref{eq:costXY}. Panel (d) shows an example for the variational optimization of a multiscale entanglement renormalization ansatz (MERA) \cite{Vidal-2005-12,Vidal2006} -- a hierarchical tensor network state that features unitary disentanglers (yellow squares) and isometries (blue triangles). See Refs.~\cite{Miao2021_08,Miao2024-109,Barthel2023_03} for more context.}
\end{figure*}

In this report, we establish a direct connection between cost-function concentration and the decay of Riemannian gradient amplitudes in the optimization of quantum circuits. The proof in the Riemannian formulation is surprisingly simple and, compared to Ref.~\cite{Arrasmith2022-7}, yields tighter bounds. We will show that when the gates are sampled according to the uniform Haar measure, the single-gate cost-function variance is exactly half the single-gate variance of the Riemannian gradient. The corresponding total variances, where all gates are varied simultaneously, are both bounded from above by sums of the single-gate variances. Furthermore, the total variances bound all individual single-gate variances. As a consequence, the barren plateau problem can be equivalently diagnosed through the analysis of cost function concentration and cannot be resolved by switching from gradient-based optimization to a gradient-free optimization \cite{Arrasmith2022-7,Arrasmith2021-5}.

\section{Cost function and Riemannian gradient}
Consider a generic quantum circuit $\hUC$ composed of some fixed unitary gates $\{\hW_1,\hW_2,\dotsc\}$ and variable unitary gates $\{\hU_1,\hU_2,\dotsc\}$ over which we optimize. Starting from a reference state $\dm_0$, the circuit prepares the state $\dm=\hUC\dm_0\,\hUC^\dag$. With an observable $\hO$, the cost function takes the form
\begin{equation}\label{eq:cost}
	E(\{\hU_i\})=\Tr\left( \hUC\dm_0\,\hUC^\dag \hO \right)
\end{equation}
With $\dm_0=\sum_{s=1}^S\dm_s\otimes|s\ket\bra s|$ and $\hO=\sum_{s=1}^S \hO_s\otimes |s\ket\bra s|$, this setup also covers the more general case
$E(\{\hU_i\})=\sum_{s=1}^S\Tr\left( \hUC\dm_s\,\hUC^\dag \hO_s \right)$ with a training set $\{(\dm_s,\hO_s)\}$ of $S$ initial states $\dm_s$ and associated measurement operators $\hO_s$ \cite{Arrasmith2022-7}.

Considering the dependence on one of the variable gates, $\hU\in\groupU(N)$, we can write the cost function in the compact form
\begin{equation}\label{eq:costXY}
	E(\hU) = \Tr(\hY\tU\hX\tU^\dag)\quad\text{with}\quad \tU := \hU\otimes \id_M
	\quad\text{and}\quad
	\hX,\hY\in\End(\CC^{N}\otimes\CC^{M})
\end{equation}
as illustrated in Fig.~\ref{fig:CostAndGrad}a, where the Hermitian operator $\hX$ on $\CC^{N}\otimes\CC^{M}$ comprises $\dm_0$, $\hY$ comprises $\hO$, and both comprise further circuit gates except $\hU$. See Fig.~\ref{fig:CostAndGrad}c for an example. As discussed in Refs.~\cite{Miao2024-109,Barthel2023_03}, expectation values $\bra\Psi|\hH|\Psi\ket$ of a Hamiltonian $\hH$ with respect to isometric tensor network states (TNS) $|\Psi\ket=\hUC|0\ket$ can also be written in the form \eqref{eq:costXY}. In this case the TNS are generated from a pure reference state $|0\ket$ by application of a quantum circuit,
and $\hU$ corresponds to one tensor of the TNS. The example of a multiscale entanglement renormalization ansatz (MERA) \cite{Vidal-2005-12,Vidal2006} is illustrated in Fig.~\ref{fig:CostAndGrad}d.

Here and in Sec.~\ref{sec:VarLocal}, we consider variation of one specific unitary gate $\hU$ such that the Riemannian manifold is just $\groupU(N)$; this is referred to as a \emph{``single-gate''} variation. The extension to variation of all gates (\emph{``total''} variation) on the full manifold \eqref{eq:Mfull} will be discussed in Sec.~\ref{sec:VarTotal}. Projecting the gradient $\hd=\partial_{\hU}E(\hU)=2\Tr_M(\hY\tU\hX)$ of the cost function \eqref{eq:costXY} onto the tangent space of the unitary group $\groupU(N)$ at $\hU$, we obtain the Riemannian gradient
\begin{equation}\label{eq:Riem_grad}
	\hg(\hU) = \Tr_M (\hY\tU\hX - \tU\hX\tU^\dag\hY\tU)
\end{equation}
as illustrated in Fig.~\ref{fig:CostAndGrad}b.
Given that we need to stay on the manifold $\groupU(N)$ during the optimization, $\hg$ is the relevant direction of change. As discussed in Refs.~\cite{Miao2021_08,Wiersema2023-107}, it can be efficiently measured on quantum computers. Here and in the following, $\Tr_N$ and $\Tr_M$ denote the partial traces over the first and second components of $\CC^{N}\otimes\CC^{M}$, respectively.

Let us summarize the derivation of Eq.~\eqref{eq:Riem_grad}: The $N\times N$ unitary gates are embedded in the $2N^2$ real Euclidean space $\mc{E}=\End(\CC^{N})\simeq\RR^{2N^2}$. The gradient in this embedding space is $\hd$. Using the (Euclidean) metric $(\hU,\hU'):=\operatorname{Re}\Tr(\hU^\dag\hU')$ for the embedding space $\mc{E}$, $\hd$ fulfills $\partial_\veps E(\hU+\veps\hV)|_{\veps=0}=(\hd,\hV)$ for all $\hV$. For Riemannian optimization algorithms, one needs to project $\hd$ onto the tangent space $\mc{T}_{\hU}$ of $\groupU(N)$ at $\hU$, and then construct retractions for line search, and vector transport to form linear combinations of gradient vectors from different points on the manifold \cite{Smith1994-3,Huang2015-25,Miao2021_08}. An element $\hV$ of the tangent space $\mc{T}_{\hU}$ needs to obey $(\hU+\veps \hV)^\dag(\hU+\veps \hV)=\id+\mc{O}(\veps^2)$ such that
\begin{equation}\label{eq:T}
	\mc{T}_{\hU}=\{\ui\hU\heta\,|\,\heta=\heta^\dag\in\End(\CC^{N})\}.
\end{equation}
The projection $\hg$ of $\hd$ onto this tangent space obeys $(\hV,\hg)=(\hV,\hd)$ for all $\hV\in\mc{T}_{\hU}$. This gives the Riemannian gradient $\hg = (\hd - \hU\hd^\dag\hU)/2$ which results in Eq.~\eqref{eq:Riem_grad}.

\section{Single-gate Haar-measure variances}\label{sec:VarLocal}
To evaluate averages and variances over $\groupU(N)$ (or, more generally the manifold $\M$), we employ Haar-measure integrals. The average of the Riemannian gradient \eqref{eq:Riem_grad} is zero,
\begin{equation}\label{eq:gAvg}
	\Avg \hg:= \int_{\groupU(N)} \ud U\, \hg(\hU) = \frac{1}{2}\int_{\groupU(N)} \ud U\, [\hg(\hU) + \hg(-\hU)] = 0,
\end{equation}
because $\hg$ is an odd function in $\hU$.
For the evaluation of $\Avg E$ and the variances, we only need the first and second-moment Haar-measure integrals over the unitary group. From the Weingarten formulas \cite{Weingarten1978-19,Collins2006-264} for the first and second moments, one obtains \cite{Barthel2023_03}
\begin{subequations}\label{eq:G}
\begin{align}\label{eq:G1}
	&\textstyle\Avg \hU^\dag\otimes\hU = \frac{1}{N}\Swap \quad\text{and}\\
	\label{eq:G2}
	&\textstyle\Avg \hU^\dag\otimes\hU\otimes\hU^\dag\otimes\hU
	 =\frac{1}{N^2-1}\left( 1 - \frac{1}{N}\Swap_{2,4}\right)\left(\Swap_{1,2}\Swap_{3,4}+\Swap_{1,4}\Swap_{2,3}\right)
\end{align}
\end{subequations}
with $\Swap=\sum_{i,j=1}^N|i,j\ket\bra j,i|$ and $\Swap_{k,\ell}$ swaps the $k^\text{th}$ and $\ell^\text{th}$ components of $\CC^N\otimes\CC^N\otimes\CC^N\otimes\CC^N$. Graphical representations of Eqs.~\eqref{eq:G1} and \eqref{eq:G2} are shown in Figs.~\ref{fig:Weingarten-E}a and \ref{fig:Weingarten-E}c.
Weingarten formulas can be proven using the Schur-Weyl duality and the double centralizer theorem \cite{Collins2006-264}. An illustrating proof for Eq.~\eqref{eq:G1} is given in Appx.~\ref{appx:Weingarten}.

Applying the Weingarten formulas \eqref{eq:G}, we find a simple linear relation between the single-gate cost-function variance
\begin{equation}\label{eq:EvarDef}
	\Var E = \Avg E^2 - \big(\Avg E\big)^2\quad\text{over}\quad \groupU(N)
\end{equation}
and the single-gate gradient variance $\Var \hg$. With the average gradient \eqref{eq:gAvg} being zero, we can quantify the gradient variance by
\begin{equation}\label{eq:gVarDef}
	\Var \hg := \Avg \frac{1}{N}\Tr (\hg^\dag \hg)\quad\text{over}\quad \groupU(N).
\end{equation}
This definition can be motivated as follows \cite{Barthel2023_03}: As any element of the tangent space \eqref{eq:T}, the gradient $\hg$ can be expanded in an orthonormal basis of involutory Hermitian operators $\{\hsigma_k\,|\,\hsigma_k=\hsigma^\dag_k=\hsigma^{-1}_k \}$ for $\End(\CC^N)$ with $\Tr(\hsigma_k\hsigma_{k'})=N\delta_{k,k'}$. This gives the gradient in the form $\hg=\ui \hU\sum_{k=1}^{N^2}\alpha_k \hsigma_k/N$, where each $\alpha_k$ corresponds to the derivative of one rotation angle. Hence, $\Tr (\hg^\dag \hg)/N=\sum_k\alpha_k^2/N^2$, i.e., Eq.~\eqref{eq:gVarDef} coincides with the average variance of the rotation-angle derivatives \Footnote{We ignore the heterogeneity of $\Avg({\alpha^2_k})$ for different $k=1,\dotsc,N^2$ of a single gate, because the gate Hilbert-space dimension $N$ is usually system-size independent.}.

An elementary proof given in Appx.~\ref{appx:theorem1} establishes a linear relation between the single-gate cost-function variance \eqref{eq:EvarDef} and gradient variance \eqref{eq:gVarDef}.
\begin{theorem}[Exact equivalence of single-gate cost-function and gradient variances]
  \label{theorem:singleGate}
  In the Riemannian formulation, the variance of the cost function \eqref{eq:cost} is exactly half the variance of the Riemannian gradient \eqref{eq:Riem_grad} when considering the dependence on one of the unitary gates of the quantum circuit ($\hU_{j\neq i}$ fixed), i.e.,
  \begin{equation}\label{eq:equivalenceSingle}
  	\var_{\hU_i} E(\{\hU_j\})=\frac{1}{2}\var_{\hU_i}\hg_i(\{\hU_j\})\qquad\forall\ \{\hU_{j\neq i}\}.
  \end{equation}
\end{theorem}
Of course, the proportionality of these \emph{conditional} single-gate variances translates directly into a proportionality of the averaged single-gate variances (the conditional variances \eqref{eq:equivalenceSingle} averaged over all $\hU_{j\neq i}$),
\begin{equation}\label{eq:singleGateVar}
	V_i:=\avg_{\{\hU_{j\neq i}\}}\var_{\hU_i} E(\{\hU_j\})
	\stackrel{\eqref{eq:equivalenceSingle}}{=}
	\frac{1}{2}\avg_{\{\hU_{j\neq i}\}}\var_{\hU_i}\hg_i(\{\hU_j\}).
\end{equation}

\section{Total Haar-measure variances}\label{sec:VarTotal}
In Sec.~\ref{sec:VarLocal}, we only considered the dependence of the cost function \eqref{eq:cost} on one of the unitary gates  ($\hU$) in the circuit as well as the single-gate gradient \eqref{eq:Riem_grad}. In this section, we consider the dependence on all variable unitary gates $(\hU_1,\hU_2,\dotsc,\hU_K)\in\M$ with $\hU_i\in\groupU(N_i)$ and the corresponding total variances like 
\begin{equation}\label{eq:EvarFullDef}
	\var_{\{\hU_j\}}E\equiv\avg_{\{\hU_j\}}(E^2)-(\avg_{\{\hU_j\}}E)^2
\end{equation}
for the cost function.

The full Riemannian gradient of the cost function \eqref{eq:cost} with respect to all variable gates is simply the direct sum of the individual gradients, i.e.,
\begin{equation}
 	\hg_\full=\hg_1\oplus\hg_2\oplus\dotsb\oplus\hg_k\quad\text{with}\quad
 	\hg_i= \Tr_{M_i} (\hY_i\tU_i\hX_i - \tU_i\hX_i\tU_i^\dag\hY_i\tU_i).
\end{equation}
Here $\tU_i := \hU_i\otimes \id_{M_i}$ with $\hX_i$ and $\hY_i$ depending on the remaining gates of the circuit, and $\dm_0$ as well as $\hO$ as in Eq.~\eqref{eq:costXY}.
In extension of Eq.~\eqref{eq:gVarDef}, we define the total variance of $\hg_\full$ as 
\begin{equation}\label{eq:gVarFullDef}
	\var_{\{\hU_j\}} \hg_\full
	:= \frac{1}{K}\sum_{i=1}^K \avg_{\{\hU_j\}}\frac{1}{N_i}\Tr (\hg_i^\dag \hg_i^\pdag)
	\stackrel{\eqref{eq:gVarDef}}{=}
  	 \frac{1}{K}\sum_{i=1}^K \avg_{\{\hU_{j\neq i}\}} \var_{\hU_i} \hg_i.
\end{equation}

The following central result as proven in Appx.~\ref{appx:theorem2} is based on an analysis of covariances, the law of total variance, and Theorem~\ref{theorem:singleGate}.
\begin{theorem}[Equivalence of circuit cost-function concentration and gradient vanishing]
  \label{theorem:global}
  When averaging over the variable unitaries $\{\hU_{i}\}$ of the quantum circuit $\hUC$ according to the Haar measure, the total variance of the cost function \eqref{eq:cost} and the total variance of the full Riemannian gradient \eqref{eq:gVarFullDef} are both bounded from below by single-gate variances $V_i$ [Eq.~\eqref{eq:singleGateVar}], and they are bounded from above by or proportional to the sum $\sum_i V_i$,
  \begin{subequations}\label{eq:equivalenceTotal}
  \begin{alignat}{7}
  	\label{eq:equivalenceTotal-a}
  	V_j &\leq& \var_{\hU_1,\dotsc,\hU_K}&& E(\hU_1,\dotsc,\hU_K)
  	&\leq \sum_{i=1}^K V_i\quad&&\forall\ j\quad\text{and}\\
  	V_j &\stackrel{\eqref{eq:gVarFullDef}}{\leq}&
  	\frac{K}{2}\var_{\hU_1,\dotsc,\hU_K}&& \hg_{\mathrm{full}}(\hU_1,\dotsc,\hU_K)
  	&\stackrel{\eqref{eq:gVarFullDef}}{=}
  	\sum_{i=1}^K V_i\quad&&\forall\ j.
  	\label{eq:equivalenceTotal-b}
  \end{alignat}
  \end{subequations}
  In particular, if all single-gate variances $V_i$ of polynomial-depth circuits ($K=\operatorname{poly} n$) decay exponentially in the system size (number of qudits) $n$, then both total variances \eqref{eq:EvarFullDef} and \eqref{eq:gVarFullDef} decay exponentially in $n$. Conversely, if one of the total variances decays exponentially in $n$, then all single-gate variances also decay exponentially. So, the barren-plateau problem and exponential cost-function concentration always appear simultaneously.
\end{theorem}
Note that the conclusions below Eq.~\eqref{eq:equivalenceTotal-b} remain valid if we choose a different weighting in the definition of the full gradient variance \eqref{eq:gVarFullDef}. For example, we could also define it as $\frac{1}{\sum_{i=1}^K N_i^2}\sum_{i=1}^K N_i \avg_{\{\hU_j\}}\Tr (\hg_i^\dag \hg_i^\pdag)$, corresponding to an equal weighting of all rotation-angle derivatives in the parametrization discussed below Eq.~\eqref{eq:gVarDef}.

\section{Numerical verification}
\begin{figure}[t]
	\includegraphics[width=0.55\textwidth]{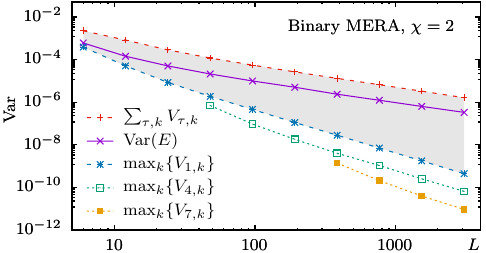}
	\caption{\label{fig:MERAboundTest} Numerical verification of Theorem~\ref{theorem:global} for binary one-dimensional MERA \cite{Vidal-2005-12,Vidal2006} with bond dimension $\chi=2$ and the cost function $E$ given by the energy expectation value \eqref{eq:MERA-E}. In accordance with Eq.~\eqref{eq:equivalenceTotal-a}, the cost-function variance is bounded from below by single-gate variances $V_{\tau,k}$ and from above by $\sum_{\tau,k} V_{\tau,k}$.}
\end{figure}
For an illustration of the general bounds on the cost-function variance in Theorem~\ref{theorem:global}, consider a one-dimensional binary MERA $|\Psi\ket$ for spin-1/2 chains of length $L=3\cdot 2^T$, where $T$ is the number of layers in the MERA \cite{Vidal-2005-12,Vidal2006}. The cost function is given by the energy (density) expectation value
\begin{equation}\label{eq:MERA-E}
	E = \bra\Psi|\hH|\Psi\ket,\quad\text{where the Hamiltonian}\quad
	\hH=\frac{1}{\sqrt{24}\,L}\sum_{i=1}^L \sum_{a=x,y,z}\hsigma_i^a \hsigma_{i+1}^a \hsigma_{i+2}^a
\end{equation}
is a sum of three-site interaction terms with Pauli operators $\hsigma^x_i,\hsigma^y_i,\hsigma^z_i$ acting on site $i$. In the evaluation of the variances, Haar-averages are executed by numerical sampling, and we denote the single-gate variances \eqref{eq:singleGateVar} by
$V_{\tau,k}$, where the position $(\tau,k)$ indicates the $k^\text{th}$ tensor in layer $\tau$.

As shown in Fig.~\ref{fig:MERAboundTest} for MERA with bond dimension $\chi=2$, the total cost-function variance $\var(E)$ [Eq.~\eqref{eq:EvarFullDef}] is, in accordance with Eq.~\eqref{eq:equivalenceTotal-a}, bounded from above by the sum of all single-gate variances $\sum_{\tau,k} V_{\tau,k}$, and the single-gate variances $V_{\tau,k}$ provide lower bounds.
Within a given layer $\tau$, the single-gate variances $V_{\tau,k}$ are approximately constant and, as shown in Refs~\cite{Miao2024-109,Barthel2023_03}, they decrease exponentially as
\begin{equation}\label{eq:MERAdecay}
	V_{\tau,k}\propto(324/625)^\tau.
\end{equation}
Hence, the best lower bound in Fig.~\ref{fig:MERAboundTest} is given by the maximum of $V_{1,k}$.

Note that the energy optimization problems for MERA, tree tensor networks states \cite{Fannes1992-66,Otsuka1996-53,Shi2006-74}, and matrix product states \cite{Fannes1992-144,Schollwoeck2011-326} for Hamiltonian with finite-range interactions are actually free of barren plateaus \cite{Miao2024-109,Barthel2023_03}. Nevertheless, the general variance relations from Theorem~\ref{theorem:global} apply.

\section{Discussion}
\begin{figure}[t]
	\includegraphics[width=0.95\textwidth]{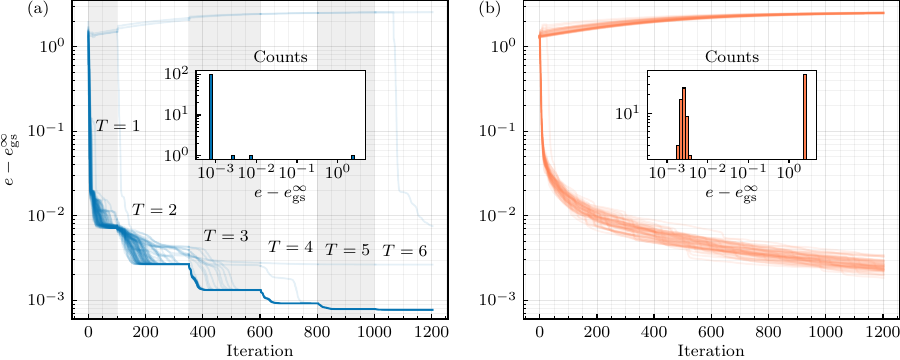}
	\caption{\label{fig:MERAbuildUp} Optimization of randomly initialized binary-MERA quantum circuits with bond dimension $\chi=2$ for the minimization of the energy expectation value of the spin-1/2 transverse-field Ising chain \eqref{eq:tIsing} at $h=1.03$. The plots show the optimization history for the deviation of the energy density $e=\bra\Psi|\hH|\Psi\ket/L$ from the exact groundstate energy density $e_\text{gs}^\infty$. Insets show histograms for the accuracy of 100 randomly initialized MERA with 6 layers after 1200 iterations. (a) The single-gate variances \eqref{eq:singleGateVar} turn out to decay exponentially in the layer index $\tau$ \cite{Miao2024-109,Barthel2023_03}. As indicated by the gray regions, we hence, first optimize layer $\tau=1$ only, then, after 100 iterations, all tensors of layers $\tau\leq 2$, then all tensors of layers $\tau\leq 3$ etc. (b) In contrast, using the traditional apporach of simultaneously optimizing all layers from the very beginning leads to considerably slower convergence and more circuits stuck in local minima.}
\end{figure}
Given the equivalence of cost-function concentration and gradient vanishing on both the single-gate as well as full-circuit levels (Theorems~\ref{theorem:singleGate} and \ref{theorem:global}, respectively), we can assess gradient vanishing and, especially, barren plateaus more easily through the scalar cost function. In fact, this route has already been pursued in recent analytic works on the trainability of variational quantum algorithms~\cite{Thanasilp2022_08,Rudolph2023_05,Ragone2023_09,Diaz2023_10,Cerezo2023_12,Xiong2023_12}.

Inspired by the work of Arrasmith \emph{et al.}\ \cite{Arrasmith2022-7} on the parametrized circuits and Euclidean gradients, we studied the question in the Riemannian formulation which makes the proofs rather simple and yields additional insights:
(a) The single-gate variances of gradients and of the cost function turn out to be strictly proportional. (b) In the Euclidean formulation, Arrasmith \emph{et al.} obtained results for the variance of cost-function \emph{differences} like $\var_{\vec{\theta}}\left(E(\vec{\theta}')-E(\vec{\theta})\right)\leq \mc{O}\left(K^2(n) V(n)\right)$, where $\vec{\theta}'$ is a random reference point, $K(n)$ is the number of variable gates as a function of the system size $n$, and $V(n)$ is a common upper bound for all single-gate (gradient) variances $V_i$. This difference construction turned out to be unnecessary in the Riemannian formulation and we could access $\var_{\hU_1,\dotsc,\hU_K}E$ directly. (c) Furthermore, we obtained the tighter bound $\var_{\hU_1,\dotsc,\hU_K}E\leq K(n)V(n)$. This result aligns with our experience in numerical simulations and could probably be further tightened.

While quantifying the total cost-function variance is easier than studying single-gate gradient variances or, equivalently, single-gate cost-function variances, the latter provide more detailed trainability information. For example, the single-gate variances in MERA tensor networks vary strongly from layer to layer. Gates in lower layers have a more substantial impact on the cost function landscape than those in upper layers. This can be taken into account to improve optimization schemes \cite{Miao2023_03}.

As a specific example, consider the optimization of the quantum circuit that defines the MERA $|\Psi\ket$ to minimize the energy expectation value $\bra\Psi|\hH|\Psi\ket$ for the spin-1/2 transverse-field Ising chain
\begin{equation}\label{eq:tIsing}
	\hH = - \sum_{i=1}^L (\hsigma_i^x\hsigma_{i+1}^x + h\hsigma_i^z).
\end{equation}
The Ising chain has a critical point at $|h|=1$, where the groundstate is particularly strongly entangled, featuring the entanglement log-area law. It follows from the analysis in Refs.~\cite{Miao2024-109,Barthel2023_03} that the total cost-function variance is (up to finite-size corrections) independent of the system size $L$ and decays algebraically with increasing MERA bond dimension $\chi$. This means that there is no barren-plateau problem and the optimization is in general possible. The single-gate variances provide more detailed information: MERA are hierarchical tensor networks with a layer structure. It turns out that the single-gate variances decay exponentially in the layer index $\tau=1,\dotsc,T$, where $T$ is the number of MERA layers \cite{Miao2024-109,Barthel2023_03}. See Eq.~\eqref{eq:MERAdecay} for binary MERA with $\chi=2$.
As demonstrated in Fig.~\ref{fig:MERAbuildUp}, this suggests a more efficient optimization scheme \cite{Miao2023_03}, where we start by setting the gates of layers $\tau\geq 2$ to $\hU_{\tau,k}=\id$ and, initially, only optimize those of layer $\tau= 1$. After a suitable number of iterations, we proceed by optimizing the gates of layers $\tau\leq 2$, then those of layers $\tau\leq 3$ and continue in this way, building up the MERA circuit layer by layer. The numerical results close to the critical point of the model confirm that this scheme is more efficient than the traditional approach of, right away, optimizing all layers simultaneously. On average, one achieves a higher energy accuracy and less circuits remain stuck in local minima.

While single-gate variances provide considerably more information than the total cost-function variance alone, they still give limited insight about trainability and convergence properties. They can show that gradient amplitudes are \emph{on average} above or below certain thresholds, but they are certainly not a measure for the complexity of the cost function landscape, the importance of local minima, or specifics of optimization trajectories.

\begin{acknowledgments}
We gratefully acknowledge helpful discussions with Baoyou Qu and support by the National Science Foundation (NSF) Quantum Leap Challenge Institute for Robust Quantum Simulation (Award No.\ OMA-2120757).
\end{acknowledgments}

\appendix

\section{Proof of the first Weingarten formula \eqref{eq:G1}}\label{appx:Weingarten}
\begin{proof}
Consider the Haar-measure integral
\begin{equation}\label{eq:B}
	\hB := \int_{\groupU(N)} \ud U\, \hU^\dag \hA \hU
\end{equation}
over the unitary group $\groupU(N)$, where $\hA,\hB\in\End(\CC^N)$ are linear operators. Due to the invariance of the Haar measure, we have
\begin{equation}
	\hV^\dag\hB\hV
	= \int_{\groupU(N)} \ud U\, (\hU\hV)^\dag \hA (\hU\hV) = \hB
	\quad\text{for all}\quad \hV\in\groupU(N).
\end{equation}
According to Schur's lemma, a linear operator $\hB$ that commutes with all elements $\hV$ of the unitary group is a scalar multiple of the identity, i.e., $\hB=\lambda \id_N$. Taking the trace, we have
\begin{equation}
	N\lambda = \Tr\hB
	\stackrel{\eqref{eq:B}}{=}
	\int_{\groupU(N)} \ud U\, \Tr(\hU^\dag \hA \hU)
	=\int_{\groupU(N)} \ud U\, \Tr(\hA) = \Tr(\hA).
\end{equation}
Now, choosing $\hA=|i\ket\bra j|$ for an orthonormal basis $\bra i|j\ket=\delta_{i,j}$, we can conclude that
\begin{equation}
	\frac{\Tr(\hA)}{N}\,\id_N = \int_{\groupU(N)} \ud U\, \hU^\dag \hA \hU
	\quad \Rightarrow\quad
	\int_{\groupU(N)}\ud U\, U^*_{i,k} U_{j,\ell}  = \frac{1}{N} \delta_{i,j}\delta_{k,\ell},
\end{equation}
which is the first Weingarten formula and equivalent to Eq.~\eqref{eq:G1}. See also Fig.~\ref{fig:Weingarten-E}a.
\end{proof}

\section{Proof of Theorem~\ref{theorem:singleGate}}\label{appx:theorem1}
\begin{figure*}[t]
	\includegraphics[width=0.97\textwidth]{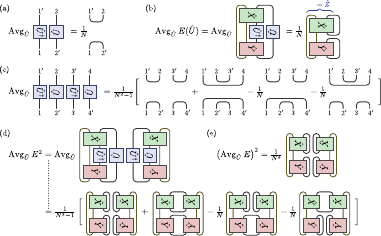}
	\caption{\label{fig:Weingarten-E}
	(a) Diagrammatic representation for the first-moment Haar-measure integral \eqref{eq:G1} over the unitary group $\groupU(N)$.
	(b) The single-gate Haar average \eqref{eq:Eavg} of the cost function \eqref{eq:costXY}.
	(c) Second-moment Haar-measure integral \eqref{eq:G2} over the unitary group $\groupU(N)$.
	(d) Average \eqref{eq:E2avg} of the squared cost function.
	(e) The single-gate Haar variance \eqref{eq:Evar} is obtained by subtracting the squared cost-function average.}
\end{figure*}
\begin{proof}
Applying the first Weingarten formula \eqref{eq:G1}, illustrated in Fig.~\ref{fig:Weingarten-E}a, the Haar-measure average of the cost-function \eqref{eq:costXY} evaluates to
\begin{equation}\label{eq:Eavg}
	\Avg E \stackrel{\eqref{eq:G1}}{=} 
	\frac{1}{N}\Tr\left(\Tr_N(\hX)\Tr_N(\hY)\right) = \frac{1}{N}\Tr \hZ
\end{equation}
as shown in Fig.~\ref{fig:Weingarten-E}b, where $\hZ\in\End(\CC^N\otimes\CC^N)$ with
\begin{equation}
	\bra i_1,i_2|\hZ|j_1,j_2\ket=\sum_{m,m'=1}^M \bra i_1,m|\hX|j_1,m'\ket\bra i_2,m'|\hY|j_2,m\ket.
\end{equation}
Applying the second Weingarten formula \eqref{eq:G2}, illustrated in Fig.~\ref{fig:Weingarten-E}c, the second moment of the cost function evaluates to
\begin{equation}\label{eq:E2avg}
	\Avg E^2 \stackrel{\eqref{eq:G2}}{=}
	 \frac{1}{N^2-1}\left((\Tr\hZ)^2-\frac{\Tr\big((\Tr_1\hZ)^2+(\Tr_2\hZ)^2\big)}{N}+\Tr\hZ^2\right),
\end{equation}
where $\Tr_1\hZ$ and $\Tr_2\hZ$ denote the partial traces of $\hZ$ over the first and second components of $\CC^N\otimes\CC^N$, respectively.
See Fig.~\ref{fig:Weingarten-E}d. Hence, the single-gate cost-function variance $\Var E = \Avg E^2 - (\Avg E)^2$ over $\groupU(N)$ is
\begin{equation}\label{eq:Evar}
	\Var E \stackrel{\eqref{eq:Eavg},\eqref{eq:E2avg}}{=}
	\frac{1}{N^2-1}\left(\frac{(\Tr\hZ)^2}{N^2}-\frac{\Tr\big((\Tr_1\hZ)^2+(\Tr_2\hZ)^2\big)}{N}+\Tr\hZ^2\right).
\end{equation}

As discussed in Ref.~\cite{Barthel2023_03} and shown diagrammatically in Fig.~\ref{fig:gradVar}, the single-gate gradient variance \eqref{eq:gVarDef} valuates to
\begin{equation}\label{eq:gVar}
	\Var\hg
	\stackrel{\eqref{eq:Riem_grad},\eqref{eq:G}}{=}
	\frac{2}{N^2-1}\left(\frac{(\Tr\hZ)^2}{N^2}-\frac{\Tr\big((\Tr_1\hZ)^2+(\Tr_2\hZ)^2\big)}{N}+\Tr\hZ^2\right)
\end{equation}
So, the cost-function variance \eqref{eq:Evar} is \emph{exactly half} of the Riemannian gradient variance \eqref{eq:gVar}. 
\end{proof}
\begin{figure*}[t]
	\includegraphics[width=0.97\textwidth]{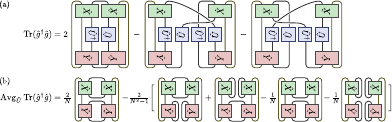}
	\caption{\label{fig:gradVar}
	(a) Diagrammatic representation of the central quantity $\Tr(\hg^\dag\hg)$ in the definition \eqref{eq:gVarDef} of the Riemannian gradient variance.
	(b) The single-gate variance \eqref{eq:gVar} of the Riemannian gradient, obtained by applying the second-moment Weingarten formula \eqref{eq:G2} as illustrated in Fig.~\ref{fig:Weingarten-E}c.}
\end{figure*}

\section{Proof of Theorem~\ref{theorem:global}}\label{appx:theorem2}
\begin{proof}
\textbf{(a)} Let us first consider a circuit with only two variable unitaries $\hU_1$ and $\hU_2$. In this case, the cost function \eqref{eq:cost} can be written in the form 
\begin{equation}\label{eq:E12-split}
	E=E(\hU_1,\hU_2)=\sum_a f_a(\hU_1)g_a(\hU_2),
\end{equation}
where $f_a$ and $g_a$ are continuous functions which only depend on $\hU_1$ and $\hU_2$, respectively: Analogously to Fig.~\ref{fig:CostAndGrad}c, we can always bipartition the the tensor network for $E(\hU_1,\hU_2)$ into two parts $f$ and $g$ with $f$ containing $\hU_1^\pdag,\hU_1^\dag$ and $g$ containing $\hU_2^\pdag,\hU_2^\dag$. The contraction of the two parts (operator products and trace to obtain the scalar $E$) then corresponds to the sum over $a$ in Eq.~\eqref{eq:E12-split}.

The single-gate cost variance for $\hU_1$ at fixed $\hU_2$ then is ($\avg_i\equiv \avg_{\hU_i}$ and $\var_i\equiv \var_{\hU_i}$)
\begin{align}\nonumber
	\var_1 E(\hU_1,\hU_2)
	&=\avg_1( E^2) - (\avg_1 E)^2\\
	&=\sum_{a,b}\underbrace{\left[\avg_1(f_a f_b)-\avg_1(f_a)\avg_1(f_b)\right]}_{\equiv \cov(f_a,f_b)}g_a g_b
\end{align}
and, similarly $\var_2 E = \sum_{a,b} f_a f_b \cov(g_a,g_b)$.

Using that $\avg_{1,2}(f_a f_b g_a g_b)=\avg_1(f_a f_b)\avg_2(g_a g_b)$ due to the independence of $\hU_1$ and $\hU_2$ in the Haar-measure average, the global cost-function variance is ($\avg_{1,2}\equiv \avg_{\hU_1,\hU_2}\equiv\avg_{\hU_1}\avg_{\hU_2}$)
\begin{align}\nonumber
	\var_{1,2} E
	&=\avg_{1,2}( E^2) - (\avg_{1,2} E)^2\\\nonumber
	&=\sum_{a,b}\left[\avg_1(f_a f_b)\avg_2(g_a g_b)-\avg_1(f_a)\avg_1(f_b)\avg_2(g_a)\avg_2(g_b)\right]\\\nonumber
	&=\avg_2\var_1 E + \avg_1\var_2 E - \sum_{a,b} \cov(f_a,f_b) \cov(g_a,g_b)\\
	&\leq \avg_2\var_1 E + \avg_1\var_2 E
	\label{eq:E12-var}
\end{align}
This is the right inequality in Eq.~\eqref{eq:equivalenceTotal-a} for the case of two variable unitaries. In the last step, we have used that the covariance matrices
\begin{equation}
	\cov(f_a,f_b) = \avg(f_a f_b)-\avg(f_a)\avg(f_b) = \avg\left([f_a-\avg f_a][f_b-\avg f_b]\right)
\end{equation}
and $\cov(g_a,g_b)$ are positive semidefinite such that the trace $\sum_{a,b} \cov(f_a,f_b) \cov(g_a,g_b)$ of their product is non-negative.

\textbf{(b)} The generalization to a circuit with $K$ variable unitaries follows by iterating Eq.~\eqref{eq:E12-var}. Decomposing the tensor network as before into $K$ parts, each containing only one of the variable unitaries and its adjoint, we can write the cost function in the form
\begin{equation}\label{eq:EK-split}
	E=E(\hU_1,\hU_2,\dotsc,\hU_K)=\sum_a f^{(1)}_a(\hU_1)f^{(2)}_a(\hU_2)\dotsb f^{(K)}_a(\hU_K).
\end{equation}
Now, iterating Eq.~\eqref{eq:E12-var}, we find
\begin{align}\nonumber
	\var_{1,2,\dotsc,K} E
	&\leq \avg_{2,\dotsc,K}\var_1 E + \avg_1\var_{2,\dotsc,K} E\nonumber\\
	&\leq \avg_{2,\dotsc,K}\var_1 E + \avg_1\left(
	  \avg_{3,\dotsc,K}\var_2 E + \avg_2\var_{3,\dotsc,K} E\right)\nonumber\\\nonumber
	&\leq \dots \leq \sum_i \avg_{\{j\neq i\}}\var_i E \equiv \sum_i V_i,
\end{align}
where $\avg_{i_1,\dotsc,i_n} h\equiv\avg_{i_1}\dots\avg_{i_n}h$ and
$\var_{i_1,\dots,i_n} h\equiv\avg_{i_1,\dotsc,i_n}h^2-(\avg_{i_1,\dotsc,i_n}h)^2$.

\textbf{(c)} The right inequality in Eq.~\eqref{eq:equivalenceTotal-a} follows by applying the law of total variance,
\begin{equation}\label{eq:EK-totalVar}
	\var_{1,2,\dotsc,K} E 
	= \avg_{\{j\neq i\}} \var_i E + \var_{\{j\neq i\}} \avg_i E
	\geq \avg_{\{j\neq i\}} \var_i E
	\stackrel{\eqref{eq:singleGateVar}}{\equiv} V_i,
\end{equation}
which holds for all $i$. Recall that, given two random variables $E$ and $U_i$ on the same probably space ($\M$), the law of total variance states that $\var(E)=\avg\big(\!\var(E|U_i)\big)+\var\big(\!\avg(E|U_i)\big)$. This corresponds to the first step in Eq.~\eqref{eq:EK-totalVar}. In the second, step, we have used the nonegativity of the variance.
\end{proof}


\begin{thebibliography}{10}

\bibitem{Cerezo2021-3}
M. Cerezo, A. Arrasmith, R. Babbush, S.~C. Benjamin, S. Endo, K. Fujii, J.~R.
  McClean, K. Mitarai, X. Yuan, L. Cincio, and P.~J. Coles, {\em Variational
  quantum algorithms}, \href{https://doi.org/10.1038/s42254-021-00348-9} {Nat.
  Rev. Phys. {\bf 3},  625  (2021)}.

\bibitem{McClean2018-9}
J.~R. McClean, S. Boixo, V.~N. Smelyanskiy, R. Babbush, and H. Neven, {\em
  Barren plateaus in quantum neural network training landscapes},
  \href{https://doi.org/10.1038/s41467-018-07090-4} {Nat. Commun. {\bf 9},
  4812  (2018)}.

\bibitem{Cerezo2021-12}
M. Cerezo, A. Sone, T. Volkoff, L. Cincio, and P.~J. Coles, {\em Cost function
  dependent barren plateaus in shallow parametrized quantum circuits},
  \href{https://doi.org/10.1038/s41467-021-21728-w} {Nat. Commun. {\bf 12},
  1791  (2021)}.

\bibitem{Grant2019-3}
E. Grant, L. Wossnig, M. Ostaszewski, and M. Benedetti, {\em An initialization
  strategy for addressing barren plateaus in parametrized quantum circuits},
  \href{https://doi.org/10.22331/q-2019-12-09-214} {Quantum {\bf 3},  214
  (2019)}.

\bibitem{Zhang2022_03}
K. Zhang, M.-H. Hsieh, L. Liu, and D. Tao, {\em Escaping from the barren
  plateau via Gaussian initializations in deep variational quantum circuits},
  \href{http://arxiv.org/abs/2203.09376} {arXiv:2203.09376  (2022)}.

\bibitem{Mele2022_06}
A.~A. Mele, G.~B. Mbeng, G.~E. Santoro, M. Collura, and P. Torta, {\em Avoiding
  barren plateaus via transferability of smooth solutions in Hamiltonian
  Variational Ansatz}, \href{http://arxiv.org/abs/2206.01982} {arXiv:2206.01982
   (2022)}.

\bibitem{Kulshrestha2022_04}
A. Kulshrestha and I. Safro, {\em BEINIT: Avoiding barren plateaus in
  variational quantum algorithms}, \href{http://arxiv.org/abs/2204.13751}
  {arXiv:2204.13751  (2022)}.

\bibitem{Dborin2022-7}
J. Dborin, F. Barratt, V. Wimalaweera, L. Wright, and A.~G. Green, {\em Matrix
  product state pre-training for quantum machine learning},
  \href{https://doi.org/10.1088/2058-9565/ac7073} {Quantum Sci. Technol. {\bf
  7},  035014  (2022)}.

\bibitem{Skolik2021-3}
A. Skolik, J.~R. McClean, M. Mohseni, P. van~der Smagt, and M. Leib, {\em
  Layerwise learning for quantum neural networks},
  \href{https://doi.org/10.1007/s42484-020-00036-4} {Quantum Mach. Intell. {\bf
  3},  5  (2021)}.

\bibitem{Slattery2022-4}
L. Slattery, B. Villalonga, and B.~K. Clark, {\em Unitary block optimization
  for variational quantum algorithms},
  \href{https://doi.org/10.1103/PhysRevResearch.4.023072} {Phys. Rev. Research
  {\bf 4},  023072  (2022)}.

\bibitem{Haug2021_04}
T. Haug and M. Kim, {\em Optimal training of variational quantum algorithms
  without barren plateaus}, \href{http://arxiv.org/abs/2104.14543}
  {arXiv:2104.14543  (2021)}.

\bibitem{Sack2022-3}
S.~H. Sack, R.~A. Medina, A.~A. Michailidis, R. Kueng, and M. Serbyn, {\em
  Avoiding barren plateaus using classical shadows},
  \href{https://doi.org/10.1103/PRXQuantum.3.020365} {PRX Quantum {\bf 3},
  020365  (2022)}.

\bibitem{Rad2022_03}
A. Rad, A. Seif, and N.~M. Linke, {\em Surviving the barren plateau in
  variational quantum circuits with Bayesian learning initialization},
  \href{http://arxiv.org/abs/2203.02464} {arXiv:2203.02464  (2022)}.

\bibitem{Tao2022_05}
Z. Tao, J. Wu, Q. Xia, and Q. Li, {\em LAWS: Look around and warm-start natural
  gradient descent for quantum neural networks},
  \href{http://arxiv.org/abs/2205.02666} {arXiv:2205.02666  (2022)}.

\bibitem{Wang2023_02}
Y. Wang, B. Qi, C. Ferrie, and D. Dong, {\em Trainability enhancement of
  parameterized quantum circuits via reduced-domain parameter initialization},
  \href{http://arxiv.org/abs/2302.06858} {arXiv:2302.06858  (2023)}.

\bibitem{Miao2024-109}
Q. Miao and T. Barthel, {\em Isometric tensor network optimization for
  extensive Hamiltonians is free of barren plateaus},
  \href{https://doi.org/10.1103/PhysRevA.109.L050402} {Phys. Rev. A {\bf 109},
  L050402  (2024)}.

\bibitem{Barthel2023_03}
T. Barthel and Q. Miao, {\em Absence of barren plateaus and scaling of
  gradients in the energy optimization of isometric tensor network states},
  \href{http://arxiv.org/abs/2304.00161} {arXiv:2304.00161  (2023)}.

\bibitem{Zhang2024-132}
H.-K. Zhang, S. Liu, and S.-X. Zhang, {\em Absence of barren plateaus in finite
  local-depth circuits with long-range entanglement},
  \href{https://doi.org/10.1103/PhysRevLett.132.150603} {Phys. Rev. Lett. {\bf
  132},  150603  (2024)}.

\bibitem{Cerezo2023_12}
M. Cerezo, M. Larocca, D. García-Martín, N.~L. Diaz, P. Braccia, E. Fontana,
  M.~S. Rudolph, P. Bermejo, A. Ijaz, S. Thanasilp, E.~R. Anschuetz, and Z.
  Holmes, {\em Does provable absence of barren plateaus imply classical
  simulability? Or, why we need to rethink variational quantum computing},
  \href{http://arxiv.org/abs/2312.09121} {arXiv:2312.09121  (2023)}.

\bibitem{Arrasmith2022-7}
A. Arrasmith, Z. Holmes, M. Cerezo, and P.~J. Coles, {\em Equivalence of
  quantum barren plateaus to cost concentration and narrow gorges},
  \href{https://doi.org/10.1088/2058-9565/ac7d06} {Quantum Sci. Technol. {\bf
  7},  045015  (2022)}.

\bibitem{Smith1994-3}
S.~T. Smith,  in {\em Hamiltonian and Gradient Flows, Algorithms, and Control},
  Vol.~3 of {\em Fields Institute Communications} (AMS, Providence, RI, 1994),
  Chap.~Optimization techniques on Riemannian manifolds, p.\ 113.

\bibitem{Huang2015-25}
W. Huang, K.~A. Gallivan, and P.-A. Absil, {\em A Broyden class of quasi-Newton
  methods for Riemannian optimization},
  \href{https://doi.org/10.1137/140955483} {SIAM Journal on Optimization {\bf
  25},  1660  (2015)}.

\bibitem{Miao2021_08}
Q. Miao and T. Barthel, {\em Quantum-classical eigensolver using multiscale
  entanglement renormalization},
  \href{https://doi.org/10.1103/PhysRevResearch.5.033141} {Phys. Rev. Research
  {\bf 5},  033141  (2023)}.

\bibitem{Wiersema2023-107}
R. Wiersema and N. Killoran, {\em Optimizing quantum circuits with Riemannian
  gradient flow}, \href{https://doi.org/10.1103/PhysRevA.107.062421} {Phys.
  Rev. A {\bf 107},  062421  (2023)}.

\bibitem{Vidal-2005-12}
G. Vidal, {\em Entanglement renormalization},
  \href{https://doi.org/10.1103/PhysRevLett.99.220405} {Phys. Rev. Lett. {\bf
  99},  220405  (2007)}.

\bibitem{Vidal2006}
G. Vidal, {\em Class of quantum many-body states that can be efficiently
  simulated}, \href{https://doi.org/10.1103/PhysRevLett.101.110501} {Phys. Rev.
  Lett. {\bf 101},  110501  (2008)}.

\bibitem{Arrasmith2021-5}
A. Arrasmith, M. Cerezo, P. Czarnik, L. Cincio, and P.~J. Coles, {\em Effect of
  barren plateaus on gradient-free optimization},
  \href{https://doi.org/10.22331/q-2021-10-05-558} {Quantum {\bf 5},  558
  (2021)}.

\bibitem{Weingarten1978-19}
D. Weingarten, {\em Asymptotic behavior of group integrals in the limit of
  infinite rank}, \href{https://doi.org/10.1063/1.523807} {J. Math. Phys. {\bf
  19},  999  (1978)}.

\bibitem{Collins2006-264}
B. Collins and P. {\'{S}}niady, {\em Integration with respect to the Haar
  measure on unitary, orthogonal and symplectic group},
  \href{https://doi.org/10.1007/s00220-006-1554-3} {Commun. in Math. Phys. {\bf
  264},  773  (2006)}.

\bibitem{Note1}
We ignore the heterogeneity of $\Avg ({\alpha ^2_k})$ for different $k=1,\dotsc
  ,N^2$ of a single gate, because the gate Hilbert-space dimension $N$ is
  usually system-size independent.

\bibitem{Fannes1992-66}
M. Fannes, B. Nachtergaele, and R.~F. Werner, {\em Ground states of {VBS}
  models on cayley trees}, \href{https://doi.org/10.1007/bf01055710} {J. Stat.
  Phys. {\bf 66},  939  (1992)}.

\bibitem{Otsuka1996-53}
H. Otsuka, {\em Density-matrix renormalization-group study of the spin-$1/2$
  $\mathrm{XXZ}$ antiferromagnet on the Bethe lattice},
  \href{https://doi.org/10.1103/PhysRevB.53.14004} {Phys. Rev. B {\bf 53},
  14004  (1996)}.

\bibitem{Shi2006-74}
Y.-Y. Shi, L.-M. Duan, and G. Vidal, {\em Classical simulation of quantum
  many-body systems with a tree tensor network},
  \href{https://doi.org/10.1103/PhysRevA.74.022320} {Phys. Rev. A {\bf 74},
  022320  (2006)}.

\bibitem{Fannes1992-144}
M. Fannes, B. Nachtergaele, and R.~F. Werner, {\em Finitely correlated states
  on quantum spin chains}, \href{https://doi.org/10.1007/BF02099178} {Commun.
  Math. Phys. {\bf 144},  443  (1992)}.

\bibitem{Schollwoeck2011-326}
U. Schollw\"{o}ck, {\em The density-matrix renormalization group in the age of
  matrix product states}, \href{https://doi.org/10.1016/j.aop.2010.09.012}
  {Ann. Phys. {\bf 326},  96  (2011)}.

\bibitem{Thanasilp2022_08}
S. Thanasilp, S. Wang, M. Cerezo, and Z. Holmes, {\em Exponential concentration
  and untrainability in quantum kernel methods},
  \href{http://arxiv.org/abs/2208.11060} {arXiv:2208.11060  (2022)}.

\bibitem{Rudolph2023_05}
M.~S. Rudolph, S. Lerch, S. Thanasilp, O. Kiss, S. Vallecorsa, M. Grossi, and
  Z. Holmes, {\em Trainability barriers and opportunities in quantum generative
  modeling}, \href{http://arxiv.org/abs/2305.02881} {arXiv:2305.02881  (2023)}.

\bibitem{Ragone2023_09}
M. Ragone, B.~N. Bakalov, F. Sauvage, A.~F. Kemper, C.~O. Marrero, M. Larocca,
  and M. Cerezo, {\em A unified theory of barren plateaus for deep parametrized
  quantum circuits}, \href{http://arxiv.org/abs/2309.09342} {arXiv:2309.09342
  (2023)}.

\bibitem{Diaz2023_10}
N.~L. Diaz, D. García-Martín, S. Kazi, M. Larocca, and M. Cerezo, {\em
  Showcasing a barren plateau theory beyond the dynamical Lie algebra},
  \href{http://arxiv.org/abs/2310.11505} {arXiv:2310.11505  (2023)}.

\bibitem{Xiong2023_12}
W. Xiong, G. Facelli, M. Sahebi, O. Agnel, T. Chotibut, S. Thanasilp, and Z.
  Holmes, {\em On fundamental aspects of quantum extreme learning machines},
  \href{http://arxiv.org/abs/2312.15124} {arXiv:2312.15124  (2023)}.

\bibitem{Miao2023_03}
Q. Miao and T. Barthel, {\em Convergence and quantum advantage of Trotterized
  MERA for strongly-correlated systems}, \href{http://arxiv.org/abs/2303.08910}
  {arXiv:2303.08910  (2023)}.

\end{thebibliography}
\end{document}